\documentclass[reqno,12pt]{amsart}
\usepackage{amsmath,amsfonts}
\usepackage{amssymb}
\usepackage{amsthm}
\usepackage[mathscr]{eucal}
\usepackage{texdraw}
\usepackage{tikz}
\usepackage{times,txfonts}
\usepackage{caption}
\usepackage{color}
\usepackage{float}
\usepackage{centernot}
\usepackage{mathtools}
\usepackage{stmaryrd}
\usepackage{enumitem}
\usepackage{longtable}

\newcommand{\be}{\begin{equation}}
\newcommand{\ee}{\end{equation}}
\newcommand{\dis}{\displaystyle}

\numberwithin{figure}{section}
\numberwithin{equation}{section}
\numberwithin{table}{section}

\theoremstyle{plain}
\newtheorem{thm}{Theorem}[section]
\newtheorem{prop}[thm]{Proposition}
\newtheorem{theorem}[thm]{Theorem}

\newtheorem{lemma}[thm]{Lemma}
\theoremstyle{definition}
\newtheorem{dfn}[thm]{Definition}
\newtheorem{example}[thm]{\it Example}
\theoremstyle{remark}
\newtheorem{remark}[thm]{\it Remark}
\usepackage{appendix}

\title{Re-factorising a QRT map}

\author{Nalini Joshi}
\address{School of Mathematics and Statistics, University of Sydney, NSW 2006, Australia}
\curraddr{}
\email{nalini.joshi@sydney.edu.au}
\thanks{}

\author{Pavlos Kassotakis }
\address{Department of Mathematics and Statistics, University of Cyprus, P.O Box: 20537, 1678 Nicosia, Cyprus}
\curraddr{}
\email{pavlos1978@gmail.com, pkasso01@ucy.ac.cy}
\thanks{}

\begin{document}
\maketitle

\begin{abstract}
A QRT map is the composition of two involutions on a biquadratic curve: one switching the $x$-coordinates of two intersection points with a given horizontal line, and the other switching the $y$-coordinates of two intersections with a vertical line. Given a QRT map, a natural question is to ask whether it allows a decomposition into further involutions. Here we provide new answers to this question and show how they lead to a new class of maps, as well as known HKY maps and quadrirational Yang-Baxter maps.
\begin{center}
\emph{Dedicated to Reinout Quispel on the occasion of his 65th birthday.}
    \end{center}
\end{abstract}

\section{Introduction}
Geometric dynamics was the subject of a small reading group in which Reinout Quispel and the first author participated as early career researchers. The group's discussions overlapped with Reinout's study of what are now famously called Quispel-Roberts-Thompson or QRT maps \cite{QRT-1988,QRT-1989}. In this paper, we consider refactorisations of these maps, which lead to unexpected connections with other classes of maps and new examples of integrable maps.

A QRT map is the composition of two involutions on a biquadratic curve (see Figure \ref{fig1}): one switching the $x$-coordinates of two intersection points with a given horizontal line, and the other switching the $y$-coordinates of two intersections with a vertical line { \cite{QRT-1988,QRT-1989,tsuda-2004,duistermaat-2010}. A natural question to ask is whether there exist further factorisations of these involutions. We answer this question in the current paper.

Let $P(x, y)$ be a biquadratic polynomial, that is, a quadratic function of each variable. Then the level set $P(x, y)=0$ defines a biquadratic curve ${\mathcal P}$ and any given horizontal line or vertical line intersects ${\mathcal P}$ at two points. The level sets of $P(x,y)$ are indexed by a free parameter $\lambda$ and the corresponding one-parameter family of curves ${\mathcal P}(\lambda)$ is called a {\em pencil} of biquadratic curves. Each curve in this pencil can be regarded as a fibre in a two-dimensional surface, which forms the phase space of the dynamical system given by the QRT map.

\subsection{Main result}
In this paper, we provide a new formulation of QRT involutions in terms of Hirota derivatives and discover conditions under which each involution can be factorized into two further involutions. These lead us to new connections between two major theories that generalize QRT maps.
The first is the class of quadrirational Yang-Baxter maps \cite{ABSyb-2004}. The second class, often called HKY maps, arises when one biquadratic curve in a pencil is mapped to another in a periodic cycle  \cite{haggar-1996,hirota-2001}. For further information and details of these two classes, as well as the history of QRT maps, we refer the reader to Section \ref{s:bg}. Furthermore, we obtain new, natural, consistent maps that satisfy the equation
\be \label{byb}
 (\mathscr{L}_{ij}\circ \mathscr{L}^{-1}_{ik} \circ \mathscr{L}_{jk})^4=id, \;\;i\neq j\neq k\in \mathbb{N},
\ee
where each $\mathscr{L}_{ij}$ is a periodic map with period $4$ and with $\mathscr{L}^{-1}_{ij}$ refers to its inverse. Throughout the paper, we use $id$ to refer to the identity operator, while $I$ is used for an invariant preserved by a map.

\subsection{Background}\label{s:bg}
Since a QRT map leaves a biquadratic pencil (and thereby its free parameter) invariant, it has a conserved quantity and is completely integrable in the sense of Liouville \cite{maeda-1987,veselov-IM-1991,santini-1991} in dynamical systems theory. By resolving the singularities of the compactified phase space,  Tsuda \cite{tsuda-2004} showed that the resulting (two-complex-dimensional) space is a rational elliptic surface fibered by the pencil ${\mathcal P}$. For a comprehensive study of the geometry of this space and QRT maps we refer to Duistermaat's book \cite{duistermaat-2010}.

The name QRT arises from an 18-parameter Liouville integrable map introduced by Quispel, Roberts and Thompson  \cite{QRT-1988,QRT-1989} in 1988.  At that time, the map was thought of as  generalisation of a second order mapping introduced in 1971 by McMillan \cite{mcmillan-1971}.
\begin{center}
\begin{figure}[H]
\begin{minipage}{6.8cm}
\begin{tikzpicture}[scale=0.60]
\draw[rotate=0] [red] (-2,-1) .. controls (-3,-0.5) and (-3,0.5) .. (-2,1) .. controls (-1.4,1.4) .. (-1,2) .. controls (-0.5,3) and (0.5,3) .. (1,2) .. controls (1.4,1.4) .. (2,1) .. controls (3,0.5) and (3,-0.5) .. (2,-1) .. controls (1.4,-1.4) .. (1,-2) .. controls (0.5,-3) and (-0.5,-3) .. (-1,-2) .. controls (-1.4,-1.4) .. (-2,-1);
 \draw[fill] (-2.5,-0.6) circle (0.08cm) ;
\node at (-2.1,-0.3) {{\tiny (x,y)}};
 \draw (-4,-0.6) -- (4,-0.6); \node at (4.1,-0.5) {{\tiny $i_1$}};
 \draw[fill] (2.5,-0.6) circle (0.08cm) ;
\node at (1.8,-0.3) {{\tiny (${\tilde x}$,y)}};
  \draw (2.5,-3.5) -- (2.5,3.5); \node at (2.2,-4) {{\tiny $i_2$}};
\draw[fill] (2.5,0.65) circle (0.08cm) ;
\node at (2.0,1.1) {{\tiny (${\tilde x}$,${\tilde y}$)}};
\end{tikzpicture}
\caption{The QRT map} \label{fig1}
\end{minipage}
\begin{minipage}{6.8cm}
\begin{tikzpicture}[scale=0.60]
\draw[rotate=60] [blue] (-2,-1) .. controls (-3,-0.5) and (-3,0.5) .. (-2,1) .. controls (-1.4,1.4) .. (-1,2) .. controls (-0.5,3) and (0.5,3) .. (1,2) .. controls (1.4,1.4) .. (2,1) .. controls (3,0.5) and (3,-0.5) .. (2,-1) .. controls (1.4,-1.4) .. (1,-2) .. controls (0.5,-3) and (-0.5,-3) .. (-1,-2) .. controls (-1.4,-1.4) .. (-2,-1);
 \draw[rotate=15] [red] (-2,-1) .. controls (-3,-0.5) and (-3,0.5) .. (-2,1) .. controls (-1.4,1.4) .. (-1,2) .. controls (-0.5,3) and (0.5,3) .. (1,2) .. controls (1.4,1.4) .. (2,1) .. controls (3,0.5) and (3,-0.5) .. (2,-1) .. controls (1.4,-1.4) .. (1,-2) .. controls (0.5,-3) and (-0.5,-3) .. (-1,-2) .. controls (-1.4,-1.4) .. (-2,-1);
 \draw[fill] (-2.7,-0.6) circle (0.08cm) ;
\node at (-2.8,-0.3) {{\tiny (x,y)}};
 \draw (-4,-0.6) -- (4,-0.6); \node at (4.1,-0.4) {{\tiny $i_1=j_1\circ j_2$}};
 \draw[fill] (2.1,-0.6) circle (0.08cm) ;
 \draw[fill] (-1.9,-0.6) circle (0.08cm) ;
\node at (-1.7,-0.3) {{\tiny $({\tilde x},y)$}};
\node at (1.7,-0.3) {{\tiny (${\hat {\tilde x}}$,y)}};
  \draw (2.1,-3.5) -- (2.1,3.5); \node at (2.3,-4) {{\tiny $i_2=k_1\circ k_2$}};
  \draw[fill] (2.1,0.2) circle (0.08cm) ;
\node at (1.5,0.5) {{\tiny (${\hat {\tilde x}}$,${\tilde y}$)}};
\draw[fill] (2.1,1.34) circle (0.08cm) ;
\node at (1.65,1.9) {{\tiny (${\hat {\tilde x}}$,${\hat {\tilde y}}$)}};
\end{tikzpicture}
\caption{A QRT re-factorisation}   \label{fig2}
\end{minipage}
\end{figure}
\end{center}
However, the iteration of maps on elliptic curves has a long history. For example, such maps arise in a study of finite order groups by Burnside in 1911\cite{burnside-1911}. Burnside's map is also mentioned by Baker
\cite{baker-1933} and it reads
\begin{equation} \label{burnside}
 (x,y)\mapsto \left(y,\frac{1-x-y}{1-x}\right).
\end{equation}
It is a period-5 symplectic map whose iterates lie on the elliptic curve
\[(1-x)(1-y)-t\, x y (x+y-1)=0,\]
where $t$ is a free (constant) parameter.
Under the translation $x\mapsto u+1,$ $y\mapsto v+1,$  \eqref{burnside} becomes a special case of the Lyness map
\[ (u,v)\mapsto \left(v,\frac{a\,v+\alpha^2}{u}\right),   \]
with $\alpha=1=a$ \cite{lyness-1942,lyness-1945}. The Lyness map with $\alpha=1$ also appears in the context of frieze patterns introduced by Coxeter \cite{coxeter-1971}.

Another example was introduced in 1959 by Mulholland and Smith \cite{mulholland-smith-1959} and independently by Scheuer and Mandel \cite{scheuer-mandel-1959} in a study of population genetics. These results together with those of Lyness were generalised by Penrose and Smith \cite{penrose-smith-1981} where they introduced a family of quadratic maps in $\mathbb{P}^2$ that are invariant on a cubic curve.

More recent studies have focused on generalizations of QRT maps. One major thread of research considered transformations from one curve to another one in the biquadratic pencil. Transformations considered were periodic in the sense that after a certain number of iterations, the initial curve was preserved. To our knowledge, the first study in this direction was by Haggar et al. \cite{haggar-1996}, in which they considered maps that admit  $k$-integrals, i.e., functions that are conserved by the $k$-th iterate of the map but not by the original map on its own. Later, Hirota et al.  \cite{hirota-2001}  discovered more examples of such mappings, later referred to as non-QRT or HKY maps \cite{kimura-2002,joshi-2006,atkinson-2008,kassotakis-20061,kassotakis-2006,kassotakis-2013}. More general maps of this sort were introduced in \cite{kassotakis-2010} by the current authors and further generalised by Roberts and Jogia \cite{rob-2015}, where they  presented a complete description of birational maps that fix one coordinate and send a pencil of biquadratic curves to another pencil of biquadratic curves.

Another major direction of related research came from the study of the quantum Yang-Baxter equation, which originates from the theory of exactly solvable models in statistical mechanics \cite{yang-1967,baxter-1982}. It reads:
\be \label{qyb}
\mathscr{R}_{ij}\; \mathscr{R}_{ik}\;  \mathscr{R}_{jk} = \mathscr{R}_{jk}\; \mathscr{R}_{ik}\;  \mathscr{R}_{ij},
\ee
where  $\mathscr{R}: V\otimes V \mapsto V\otimes V$ is a linear operator and $\mathscr{R}_{lm}$ the operator that acts as $\mathscr{R}$ on the $l$-th and $m$-th factor of the tensor product
$V\otimes V \ldots \otimes V$. For a history and the early developments of the theory see \cite{Jimbo-1989}.

By replacing $V$ with any set $X$ and the tensor product $\otimes$ with the Cartesian product $\times$, Drinfeld \cite{drinfeld-1992} introduced the {\it set theoretical version } of (\ref{qyb}). Solutions
of the latter appeared under the name {\it set theoretical solutions of the quantum Yang-Baxter equation} \cite{sklyanin-1988,etingof-1999}; note that another class of solutions appeared in \cite{viallet-1995,hietarinta-1997}. The term {\it Yang-Baxter maps} was proposed by Veselov \cite{veselov-2003} as an alternative name to Drinfeld's one. Early results on Yang-Baxter maps were provided in \cite{adler-1993,NY-1998,KNY-2002A,Nij-Maill}.

The classification of Yang-Baxter maps is a difficult and interesting problem. In the simplest case where the set $X={\mathbb{P}}^1\equiv \mathbb P^1,$  equivalence classes of Yang-Baxter maps  under the requirement of quadrirationality (see Section \ref{quadri}) were obtained in \cite{ABSyb-2004} and complemented in \cite{pap3-2010}. Note that  the Yang-Baxter maps under consideration in \cite{ABSyb-2004,pap3-2010} are involutions $\mathscr{R}^2=id,$ hence the Yang-Baxter equation (\ref{qyb}) can be cast into the form:
\be \label{ayb}
 (\mathscr{R}_{ij}\circ \mathscr{R}_{ik} \circ \mathscr{R}_{jk})^2=id, \;\;i\neq j\neq k\in \mathbb{N}.
\ee
The connection of Yang-Baxter maps with integrable partial difference equations was originated in \cite{pap2-2006} and completed in \cite{KaNie,PKMN2,PKMN3,KaNie:2018} where also the connection with higher degree integrable quad relations was revealed. Moreover, the interplay between Yang-Baxter maps and discrete integrable systems led to fruitful  results \cite{atk-2013,BAZHANOV2018509,Caudrelier_2013,Atkinson:2018,Dimakis2018,Dimakis2018ii,AtkNie,Grahovski:2016,Mikhailov2016,Kouloukas:2018,Kassotakis:2019}.

\subsection{Outline of the paper}
The structure of the paper is as follows.
In Section 2, we define the QRT map and seed involutions, and obtain a factorisation of each QRT involution into two consecutive ones for some specific parameter matrices. We consider outcomes from this factorisation, namely Liouville integrable maps, quadrirational Yang-Baxter maps and some novel maps which are solutions of the equation (\ref{byb}) in Section 3. In Section 4 we conclude this paper with a discussion.

\section{QRT and seed involutions}
The QRT mapping is defined by the composition $\phi:=i_2\circ i_i$ of two non-commuting involutions $i_1,$ $i_2,$   which each preserve the same biquadratic invariant
\[
{\dis I(x,y)=\frac{{\bf X}^T\mathscr{A}_0{\bf
Y}}{{\bf X}^T\mathscr{A}_1{\bf Y}}},
\]
where ${\bf X}$, ${\bf Y}$ are vectors ${\bf X}=(x^2,x,1)^T,$ ${\bf Y}=(y^2,y,1)^T$ and $\mathscr{A}_i,$ $i=0,1$, are $3\times 3$ matrices, given by
\[
{\dis \mathscr{A}_i=\left(\begin{array}{ccc}
\alpha_i& \beta_i &\gamma_i\\
\delta_i& \epsilon_i& \zeta_i\\
\kappa_i& \lambda_i & \mu_i
\end{array}\right)}.
\]
We recall that the QRT map bi-rationally preserves the linear pencil of bi-quadratic curves
$$
P(x,y;t):={\bf X}^T\mathscr{A}_0{\bf Y}-t\;{\bf X}^T\mathscr{A}_1{\bf Y}.
$$
The more general case of curve dependant maps that preserve a general biquadratic foliation ie. $P(x,y;t):={\bf X}^TA(t){\bf Y},$ was studied in detail in a series of papers by Roberts et.al. \cite{Iatrou-2001,Iatrou-2002,Pettigrew-2008}.
If the matrices $\mathscr{A}_i,$ $i=0,1,$ are symmetric, i.e., $\mathscr{A}_i=\mathscr{A}_i^T$, the QRT mapping is called {\it symmetric}, due to its invariant being symmetric under the interchange of $x$ and $y$. If they are antisymmetric, i.e., $\mathscr{A}_i=-\mathscr{A}_i^T$, the QRT is called {\it antisymmetric}, otherwise it is called {\it asymmetric}.

The involutions $i_1:(x,y)\mapsto (X,Y)$ and $i_2: (x,y)\mapsto (X,Y),$ are defined by the solution of the equations $I(X,y)-I(x,y)=0$ and $I(x,Y)-I(x,y)=0$ respectively, which each have 2 solutions. One solution is the identity map and, if we omit these trivial solutions, it was shown by Quispel et al. \cite{QRT-1988,QRT-1989} that the other solution is given by
\begin{subequations}\label{QRT}
\begin{equation}\label{QRTa}
i_1:(x,y)\mapsto (X,Y),\;\;\left\{ \begin{array}{l}
X={\dis \frac{f_1(y)-f_2(y)x}{f_2(y)-f_3(y)x}}\\[3mm]
  Y=y
    \end{array}\right.,\;\;
i_2:(x,y)\mapsto (X,Y),\;\;\left\{\begin{array}{l}
X=x\\
  Y={\dis\frac{g_1(x)-g_2(x)y}{g_2(x)-g_3(x)y}}
    \end{array}\right.,
    \end{equation}
where
\begin{equation}\label{QRTb}
\begin{pmatrix}
f_1(y)\\
f_2(y)\\
f_3(y)
\end{pmatrix}
=(\mathscr{A}_0{\bf Y})\times(\mathscr{A}_1{\bf Y}),\qquad
\begin{pmatrix}
g_1(x)\\
g_2(x)\\
g_3(x)
\end{pmatrix} =(\mathscr{A}_0^T{\bf X})\times(\mathscr{A}_1^T{\bf X}).
\end{equation}
\end{subequations}
We found an alternative representation of these involutions, given by
\begin{equation} \label{h-qrt}
i_1:(x,y)\mapsto (X,Y),\;\left\{ \begin{array}{l}
{\displaystyle X}=x-2{\displaystyle \frac{D_x\; n\cdot d}{\partial_x\; D_x\; n\cdot d}}\\
Y=y
\end{array}\right.,\;\;
i_2:(x,y)\mapsto (X,Y),\;\left\{ \begin{array}{l}
X=x\\
{\displaystyle Y}=y-2{\displaystyle \frac{D_y\; n\cdot d}{\partial_y\; D_y\; n\cdot d}},
\end{array}\right.
\end{equation}
where
\begin{equation}\label{eq:nm}
n={\bf X}^T\mathscr{A}_0{\bf Y}, \quad
d={\bf X}^T\mathscr{A}_1{\bf Y},
\end{equation}
are respectively the numerator and the denominator of the QRT invariant $I(x,y)$. Moreover, $\partial_x\equiv \frac{\partial }{\partial x},$ $\partial_y\equiv \frac{\partial }{\partial y}$ are the usual partial differentiation operators, while $D_x,$ $D_y$ are the Hirota bilinear operators i.e.
\[ D_x \;f \cdot g=f_x\;g-f\;g_x, \quad D_y \;f \cdot g=f_y\;g-f\;g_y.\]
This representation of the QRT involutions is related to the affine case of  the so-called curve-dependent McMillan maps developed by Iatrou and Roberts in \cite{Iatrou-2001,Iatrou-2002}.

We now define the terminology used in the results below.
\begin{dfn}
The terms base point, singular point, QRT involution, and seed involution are defined as follows.
\begin{enumerate}[leftmargin=1.4cm,label={\rm (\roman*)}]
\item {\bf Base points:}
Points  $(x,y)$ that are contained in all curves of the pencil of biquadratic curves $P(x,y;t)$ are called \emph{base points}. For QRT maps, a base point is given by
\be \label{base-points}
x=\frac{f_1(y)}{f_2(y)}=\frac{f_2(y)}{f_3(y)},\quad \mbox{or} \quad y=\frac{g_1(x)}{g_2(x)}=\frac{g_2(x)}{g_3(x)}.
\ee
\item[(ii)] {\bf Singular points of the QRT map:}
Points $(x,y)$ where the QRT map is not defined, i.e., both numerators and denominators of the QRT map are zero or infinity, will be referred to as \emph{singular points}. For the QRT map, singular points correspond exactly to the base points of the pencil of bi-quadratic curves preserved by the map.
\item[(iii)] {\bf {\it QRT} involutions:}
The involutions $i_1, i_2$, defined in Equation (\ref{h-qrt}), are called QRT involutions.
\item[(iv)] {\bf Seed involutions:}
Define the involutions:
\begin{equation}\label{eq:seed}
    \begin{split}
&j_1:(x,y)\mapsto (X,Y),\;\left\{ \begin{array}{l}
{\displaystyle X}=x-2{\displaystyle \frac{n}{\partial_x\; n}}\\
Y=y
\end{array}\right.,\;\;
j_2:(x,y)\mapsto (X,Y),\;\left\{ \begin{array}{l}
{\displaystyle X}=x-2{\displaystyle \frac{d}{\partial_x\; d}}\\
Y=y
\end{array}\right., \\
&
k_1:(x,y)\mapsto (X,Y),\;\left\{ \begin{array}{l}
X=x\\
{\displaystyle Y}=y-2{\displaystyle \frac{n}{\partial_y\; n}}
\end{array}\right., \;\;
k_2:(x,y)\mapsto (X,Y),\;\left\{ \begin{array}{l}
X=x\\
{\displaystyle Y}=y-2{\displaystyle \frac{d}{\partial_y\; d}}
\end{array}\right.,
\end{split}
\end{equation}
where $n$ and $d$ are given in Equation \eqref{eq:nm}. The involutions $j_1, j_2, k_1, k_2$ will be referred to as \emph{ seed involutions}.
\end{enumerate}
\end{dfn}

We impose the following condition on the seed involutions:
\be\label{cond}
\begin{split}
(j_1\circ j_2)^n&=id\quad\textrm{and/or}\quad (k_1\circ k_2)^n=id,\;\;\textrm{for some}\quad  n\geq 2 \in \mathbb N ,\\
j_1\circ j_2&\not=id\quad \textrm{and}\phantom{/or}\quad  k_1\circ k_2\not=id.
\end{split}
\ee
\begin{remark}
The seed involutions as they stand they do not appear to be related to integrability. Below we show that the imposition of the condition above leads to an integrability structure under certain conditions, which in turn restrict the parameter space.
\end{remark}

In the following lemma we present the necessary conditions for the seed involutions to satisfy the relations above for $n=2,3,4,5.$ The method can be extended to arbitrary $n$.
\begin{lemma} \label{lamma1}
The seed involutions defined in Equation \eqref{eq:seed} satisfy Equation \eqref{cond} for $n=2,3,4,5,$ if the following respective conditions  hold:\\
\begin{center}
\begin{tabular}{|c|c|c|}
  \hline
  &&\\
   & $(j_1\circ j_2)^n=id$ & $(k_1\circ k_2)^n=id$ \\[10pt]
   \hline
  $n=2$ & $m=0$ & $M=0$ \\[6pt]
  $n=3$ & $4m^2-kl=0$ & $4M^2-KL=0$ \\[6pt]
  $n=4$ & $m(2m^2-kl)=0$ & $M(2M^2-KL)=0$ \\[6pt]
  $n=5$ & $16m^4-12klm^2+(kl)^2=0$ & $16M^4-12KLM^2+(KL)^2=0$ \\[6pt]
    \hline
\end{tabular}
\end{center}
where
\begin{equation} \label{dfnre}
    \begin{aligned}
&m=-n_xd_x-D^2_x n\cdot d, & M=-n_yd_y-D^2_y n\cdot d,& & k=-n n_{xx}+D_x n\cdot n_x,\\
&K=-n n_{yy}+D_y n\cdot n_y,& l=-d d_{xx}+D_x d\cdot d_x,&& L=-d d_{yy}+D_y d\cdot d_y.
\end{aligned}
\end{equation}
Here $D^2_x, D^2_y$ are Hirota operators, defined for given non-negative integers $p$, $q$ by
\[
D^p_x D^q_y\ f\cdot g = (\partial_x-\partial_{x'})^p(\partial_y-\partial_{y'})^q f(x, y)g(x', y')\Bigm|_{x'=x, y'=y}.
\]
Furthermore, subscripts on $n$ and $d$ denote partial derivatives in the usual way.
\end{lemma}
\begin{proof}
We give a proof for the case $n=2$. The remaining cases may be proved in the same way.

Observe that the composition of the seed involutions $j_1$ and $j_2$ leads to
\begin{equation*}
\begin{split}
&j_1\circ j_2: (x,y)\mapsto (X,Y),\\
&\mbox{where}\;\;\begin{cases}
\ {\displaystyle X}&=x-2{\displaystyle \frac{D_x \; d\cdot n}{\partial_x\; D_x \;d\cdot n-n_xd_x-D^2_x\; n\cdot d}}=x-2{\displaystyle \frac{D_x \; d\cdot n}{\partial_x\; D_x \;d\cdot n+m}},\\
\ Y&=y,
\end{cases}
\end{split}
\end{equation*}
with $m:=-n_xd_x-D^2_x\; n\cdot d$. Moreover, we have
\begin{equation*}
\begin{split}
&(j_1\circ j_2)^2: (x,y)\mapsto (X,Y),\\
&\mbox{where}\;\;\begin{cases}
\ {\displaystyle X}&=x-4m{\displaystyle \frac{ D_x \; d\cdot n}{kl-2m^2+2m D_x \; d\cdot n}},\\
\ Y&=y,
\end{cases}
\end{split}
\end{equation*}
with $k,l,m$ as defined in \eqref{dfnre}. Clearly a necessary condition for $(j_1\circ j_2)^2=id,$ is $m:=-n_xd_x-D^2_x n\cdot d=0.$
(The remaining possibility $D_x d\cdot n=0$ is not included, as it leads to a violation of Condition \eqref{cond}.)
Moreover, for $m=0,$ $j_1\circ j_2$ becomes the QRT involution  (\ref{h-qrt}). Working similarly, from the seed involutions $k_1,$ $k_2$, we find that a necessary condition for $(k_1\circ k_2)^2=id,$ is  $M:=-n_yd_y-D^2_y n\cdot d=0.$
\end{proof}
\begin{remark}
The results of Lemma \ref{lamma1} lead to constraints on the parametric matrices $A_0$ and $A_1.$ The seed involutions associated to these restricted matrices inherit more structure, i.e., for $n=2$, the function $I=n/d$ is a semi-invariant of the seed involutions $j_i, k_i,\;i=1,2$. Note that a function $I$ is called a {\em semi-invariant} of a map $\phi,$ if it satisfies $I\circ\phi=-I.$
\end{remark}
For the remainder of this paper we restrict our attention to the case $n=2$. We now determine the resulting constraints on the parametric matrices that follow from the case $n=2$ of Lemma  (\ref{lamma1}) and discuss the integrable structure inherited by the seed involutions.
\begin{lemma} \label{prop00}
The conditions $m=M=0$, where $m, M$ are given in (\ref{cond}), lead to
\be \label{matrices00}
{\dis \mathscr{A}_0=\left(\begin{array}{ccc}
\alpha& \beta &\gamma\\
\delta& \epsilon& \zeta\\
\kappa& \lambda & \mu
\end{array}\right)},\quad
{\dis \mathscr{A}_1=\left(\begin{array}{ccc}
|\mathscr{A}_0|_{33}& 2|\mathscr{A}_0|_{32} &|\mathscr{A}_0|_{31}\\
2|\mathscr{A}_0|_{23}& 4|\mathscr{A}_0|_{22} &2|\mathscr{A}_0|_{21}\\
|\mathscr{A}_0|_{13}& 2|\mathscr{A}_0|_{12} &|\mathscr{A}_0|_{11}
\end{array}\right)},
\ee
where $|\mathscr{A}_0|_{ij}$ is the $(ij)-$th minor determinant of the matrix $\mathscr{A}_0$.
\end{lemma}
\begin{proof}
The conditions $m=0$ and $M=0$ read respectively:
\be \label{n=2-cond}
\begin{split}
&m=-n_xd_x-D^2_x n\cdot d=bb_1-2(ac_1+a_1c)=0,\\
&M=-n_yd_y-D^2_y n\cdot d=BB_1-2(AC_1+A_1C)=0,
\end{split}
\ee
where
$$
(a,b,c)^T=\mathscr{A}_0{\bf Y}, \; (a_1,b_1,c_1)^T=\mathscr{A}_1{\bf Y}, \;(A,B,C)={\bf X}^T\mathscr{A}_0, \;(A_1,B_,C_1)={\bf X}^T\mathscr{A}_1.
$$
$n={\bf X}^T\mathscr{A}_0{\bf Y}$ and $d={\bf X}^T\mathscr{A}_1{\bf Y}.$
Equations (\ref{n=2-cond})  are quartic polynomials in $y$  and $x$  respectively. Equating their coefficients to zero leads to $8$ equations  that involve the entries of $\mathscr{A}_0, \mathscr{A}_1.$ This system of $8$ equations is linear with respect to the entries of the matrix $\mathscr{A}_0$ or of the matrix $\mathscr{A}_1$.
By considering the entries of the matrix $\mathscr{A}_1$ as unknowns, 
and solving the resulting linear system leads exactly to (\ref{matrices00}).
\end{proof}
Note that in this setting (\ref{n=2-cond}) arose as commutation conditions of the seed involutions. Conditions (\ref{n=2-cond}) also coincide with formulae (30) of Theorem 2 of \cite{rob-2015} if one considers the source and the target biquadratics to be the same  but with fibre values of opposite sign. 
At the same time, if we demand the  QRT map to factorise as the product $j_1\circ j_2\circ k_1\circ k_2$ of the seed involutions, we end up with the same conditions.   This factorisation was also a cornerstone in the construction of  HKY maps in \cite{kassotakis-2010}, since it served a perfect-square discriminantal condition (see point (ii) of the proof of Theorem \ref{theorem1} below).

The following theorem gives properties of the QRT involutions, seed involutions, and the base points of the pencil $P(x,y;t)=0,$ for the case of the restricted parameter matrices of Lemma (\ref{prop00}).
\begin{theorem} \label{theorem1}
For the integral ${\dis I(x,y)={\bf X}^T\mathscr{A}_0{\bf
Y}/{\bf X}^T\mathscr{A}_1{\bf Y}},$  or equivalently for the pencil of bi-quadratic curves $
P(x,y;t):={\bf X}^T\mathscr{A}_0{\bf Y}-t\;{\bf X}^T\mathscr{A}_1{\bf Y}
$  where ${\bf X}$, ${\bf Y}$ are vectors ${\bf X}=(x^2,x,1)^T,$ ${\bf Y}=(y^2,y,1)^T$ and $\mathscr{A}_0,$ $\mathscr{A}_1$ are the parameter matrices of Lemma (\ref{prop00}), the following results hold.
\begin{enumerate}[leftmargin=1.4cm,label={\rm (\roman*)}]
\item The QRT involutions $i_1,$ $i_2$ given by (\ref{QRT})  preserve $I(x,y)$.
\item The seed involutions $j_1,\; j_2,\; k_1,\; k_2,$ given below, anti-preserve $I(x,y)$. Here  $j_1, j_2, k_1, k_2: (x,y)\mapsto (X,Y)$ are defined by
 $$
\begin{array}{l}
 j_1:\left\{ \begin{array}{l}
X= -{\dis \frac{2(\kappa y^2+\lambda y+\mu)+(\delta y^2+\epsilon y+\zeta)x}{\delta y^2+\epsilon y+\zeta+2(\alpha y^2+\beta y+\gamma)x}} ,\\[3mm]
   Y=y
    \end{array}\right.\\
j_2:\left\{\begin{array}{l}
X= -{\dis \frac{2(|\mathscr{A}_0|_{13} y^2+|\mathscr{A}_0|_{12} y+|\mathscr{A}_0|_{11})+(|\mathscr{A}_0|_{23} y^2+|\mathscr{A}_0|_{22} y+|\mathscr{A}_0|_{21})x}{|\mathscr{A}_0|_{23} y^2+|\mathscr{A}_0|_{22} y+|\mathscr{A}_0|_{21}+
2(|\mathscr{A}_0|_{33} y^2+|\mathscr{A}_0|_{32} y+|\mathscr{A}_0|_{31})x}} ,\\ [3mm]
   Y=y
    \end{array}\right.
    \end{array}
$$
 $$\begin{array}{l}
k_1:\left\{ \begin{array}{l}
X=x ,\\[3mm]
Y= -{\dis \frac{2(\gamma x^2+\zeta x+\mu)+(\beta x^2+\epsilon x+\lambda)y}{\beta x^2+\epsilon x+\lambda+2(\alpha x^2+\delta x+\kappa)y}} ,
    \end{array} \right.\\
k_2:\left\{\begin{array}{l}
X=x ,\\ [3mm]
Y= -{\dis \frac{2(|\mathscr{A}_0|_{31} x^2+|\mathscr{A}_0|_{21} x+|\mathscr{A}_0|_{11})+(|\mathscr{A}_0|_{32} x^2+|\mathscr{A}_0|_{22} x+|\mathscr{A}_0|_{12})y}{|\mathscr{A}_0|_{32} x^2+|\mathscr{A}_0|_{22} x+|\mathscr{A}_0|_{12}+
2(|\mathscr{A}_0|_{33} x^2+|\mathscr{A}_0|_{23} x+|\mathscr{A}_0|_{13})y}} .
     \end{array} \right.
     \end{array}
$$

\item Let the set $\Sigma=\{P_1,\ldots, P_8\}$ consist of the base points of $P_i=(x_i,y_i)\in \mathbb{P}^1\times \mathbb{P}^1$ of the pencil of biquadratic curves
 $P(x,y;t)=0.$ Then $\Sigma=\sigma_1 \cup \sigma_2$,  where $\sigma_1=\{ B_1,\ldots B_4 \}, \sigma_2=\{ b_1,\ldots b_4 \}$, which are defined by
$B_i=(f_1(Y_i)/f_2(Y_i),Y_i)$ and
$b_i=(f_1(y_i)/f_2(y_i),y_i)$, where $Y_i$ are the roots of the quartic polynomial
\[
\begin{split}
h&= y^4 (\delta^2 - 4 \alpha \kappa) +
 y^3 (2 \delta \epsilon - 4 \beta \kappa -
    4 \alpha \lambda)  \\
    &\qquad +
 y^2 (\epsilon^2 + 2 \delta \zeta - 4 \gamma \kappa -
    4 \beta \lambda - 4 \alpha \mu) +
 y (2 \epsilon \zeta - 4 \gamma \lambda -
    4 \beta \mu) +\zeta^2- 4 \gamma \mu
\end{split}
\]
   while $y_i$ are the roots of the quartic polynomial
\[
\begin{split}
H&=y^4 (\delta_1^2 - 4 \alpha_1 \kappa_1) +
 y^3 (2 \delta_1 \epsilon_1 - 4 \beta_1 \kappa_1 -
    4 \alpha_1 \lambda_1)  \\
    &\qquad  +
y^2 (\epsilon_1^2 + 2 \delta_1 \zeta_1 - 4 \gamma_1 \kappa_1 -
    4 \beta_1 \lambda_1 - 4 \alpha_1 \mu_1)  \\
    &\qquad + y (2 \epsilon_1 \zeta_1 - 4 \gamma_1 \lambda_1 -
    4 \beta_1 \mu_1) +\zeta_1^2- 4 \gamma_1 \mu_1
\end{split}
\]
where $\alpha_1, \beta_1, \ldots \mu_1$ are entries of the matrix $\mathscr{A}_1$.
\end{enumerate}
\end{theorem}
\begin{proof} The first statement is well-known and we include it for compeleteness.
\begin{enumerate}[leftmargin=1.4cm,label={\rm (\roman*)}]
\item The solution of $I(X,y)-I(x,y)=0$ and $I(x,Y)-I(x,y)=0,$ apart from the trivial solutions $X=x$ and $Y=y,$  give respectively the QRT involutions $i_1, i_2.$
\item The invariant can be written as
\[ {\dis I(x, y)=\frac{{\bf X}^T\mathscr{A}_0{\bf Y}}{{\bf X}^T\mathscr{A}_1{\bf Y}}=\frac{a(y)x^2+b(y)x+c(y)}{a_1(y)x^2+b_1(y)x+c_1(y)}=\frac{A(x)y^2+B(x)y+C(x)}{A_1(x)y^2+B_1(x)y+C_1(x)} }.\]
The quadratic equation  $I(X,y)+I(x,y)=0$ has two rational solutions for $X$ when its discriminant is a perfect square. This is true when \be \label{c1}
 b b_1=2(a_1c+ac_1).
\ee
 Similarly for the equation $I(x,Y)+I(x,y)=0$, we obtain the condition \be \label{c2}
BB_1=2 (A_1 C+A C_1).
\ee
Relations (\ref{c1}) and (\ref{c2}) are exactly the conditions (\ref{n=2-cond}), so from Lemma  (\ref{prop00}) together with the definition of seed involutions, we obtain exactly the desired involutions stated in the theorem.
\item Because of the special form of the parametric matrices (\ref{matrices00}), using the first  formula of (\ref{base-points}) leads to the set of singular points $\{(x_s,y_s)\}$ of the QRT map being $\{(\frac{f_1(y_s)}{f_2(y_s)},y_s)\}$, where $y_s$ is a solution of the polynomial $h H=0.$ (The use of the second formula gives exactly the same results.)
\end{enumerate}
\end{proof}

\begin{remark}
When the source and the target biquadratics are the same  but with fibre values that are of opposite sign, the involutions $ k_i, i=1,2,$ could also be seen as the involutions $P^+, P^-$ of \cite{rob-2015} together with their analogues $ j_i, i=1,2,$ when one considers a fixed $y$. We note that the involutions  $j_i,  i=1,2,$ first appeared in the literature in \cite{kassotakis-2010} and both $j_i, k_i, i=1,2$ were presented by one of us  in the SIDE 9 conference in Varna 2010.
The defining polynomial of the base points of a pencil of biquadratic curves or equivalently of the singular points of the associated QRT map, is a degree 8 polynomial (see (\ref{base-points})). For the specific pencil of biquadratic curves defined in the theorem above, this degree 8 polynomial factorises to the product of two degree 4 polynomials. Namely the polynomials $h$ and $H.$ 
\end{remark}

In the following proposition, we refer to fixed points.  Recall that a point $(x_0,y_0)$ is a fixed point of a mapping $\phi$,  if   $\phi(x_0,y_0)=(x_0,y_0)$. Item (ii) of this proposition points out that the pair $\{j_1, k_2\}$ of seed involutions share fixed points and base points and similarly for $\{j_2, k_1\}$. However, the role of these points are interchanged between the pairs. In Item (iii), we show that  these seed involutions also form a group, which is a subgroup of a triangle Coxeter group.
\begin{prop} \label{prop1}
For the seed involutions $j_i,\; k_i\; i=1,2,$  the following results hold.
\begin{enumerate}[leftmargin=1.4cm,label={\rm (\roman*)}]
\item  $j_1\circ j_2=i_1,$\; $k_1\circ k_2=i_2,$ where $i_1, i_2$ are the QRT involutions associated with the parameter matrices $\mathscr{A}_0,\; \mathscr{A}_1$ given  in Lemma (\ref{prop00}). Hence the QRT mapping $\phi=i_1\circ i_2=j_1\circ j_2\circ k_1 \circ k_2$ refactorises as the product of seed involutions.
\item  The seed involutions $j_1, k_2$  have singularities given by the elements of the set $\sigma_1$ and fixed points given by elements of the set $\sigma_2$. On the other hand, $j_2, k_1$ have elements of $\sigma_1$ as fixed points and elements of $\sigma_2$ as singularities.
\item The seed involutions and their compositions form a  group:
\[
G=\left\{ j_1,j_2,k_1,k_2 \ \bigm|\  j_i^2=k_i^2=(j_1\circ j_2)^2=(k_1\circ k_2)^2=(j_{i}\circ k_{i+1})^2=id, \; i=1,2 \right\}.
\]
 Furthermore, $G\subset G', $ where
 \[
 \begin{split}
 G'&=\Bigl\{ J, R, S \ \bigm|\ J^2= R^2= S^2=(R\circ S)^2=(J\circ S)^4=(J\circ R\circ S)^4\\
 &\phantom{=\Bigl\{ J, R, S \ \bigm|\ J^2= R^2= S^2=(R\circ S)^2} =(S\circ R \circ J\circ R)^4=id  \Bigr\} ,
 \end{split}
 \]
 where
    $$
    \begin{array}{l}
    J:\left(x,y; \mathscr{A}_0\right)\mapsto \left(-{\dis \frac{2(\kappa y^2+\lambda y+\mu)+(\delta y^2+\epsilon y+\zeta)x}{\delta y^2+\epsilon y+\zeta+2(\alpha y^2+\beta y+\gamma)x}},y ;\mathscr{A}_0 \right),\\[3mm]
    R: \left(x,y; \mathscr{A}_0\right)\mapsto \left(y,x ; \mathscr{A}_0^T  \right),\\ [3mm]
    S: \left(x,y; \mathscr{A}_0\right)\mapsto \left(y,x ; \mathscr{A}_2 \right),\\ [3mm]
    \end{array}
    $$
     where
$$
{\dis \mathscr{A}_2=\frac{1}{4 det(\mathscr{A}_0)}\left(\begin{array}{ccc}
|\mathscr{A}_0|_{33}& 2|\mathscr{A}_0|_{32} &|\mathscr{A}_0|_{31}\\
2|\mathscr{A}_0|_{23}& 4|\mathscr{A}_0|_{22} &2|\mathscr{A}_0|_{21}\\
|\mathscr{A}_0|_{13}& 2|\mathscr{A}_0|_{12} &|\mathscr{A}_0|_{11}
\end{array}\right).}
    $$
Note that the group $ G'$ is the hyperbolic $(2,4,\infty)$ triangle Coxeter group modulo the relation $(R\circ S)^2=(J\circ S)^4=(J\circ R\circ S)^4=(S\circ R \circ J\circ R)^4=id$.
\item  The generators $R,J,S$ act on the invariant $I$ as follows:
\[ I\circ R= I,\ I\circ J= -I,\ I\circ S= 1/I.\]
\item  The seed involutions satisfy
\[ j_1=J,\ j_2=R\circ S\circ J \circ R \circ S,\ k_1=R\circ J \circ R,\ k_2=S\circ J \circ S.\]
\item We have two birational representations of the $BC_2$ Coxeter group, namely, the subgroups of $G'$, $A=\{ J,  S  | J^2=S^2=(J\circ S)^4=id\}$ and $B=\{ S,  R\circ J \circ R  | S^2=(R\circ J \circ R)^2=(S \circ R\circ J \circ R)^4=id\}.$
\item  The maps $\Omega_1= j_1\circ k_2=(J\circ S)^2$ and $\Omega_2= j_2\circ k_1= (S \circ R\circ J \circ R)^2$ are elements  of $BC_2,$ arising from  two different representations of the  latter. Moreover they are quadrirational maps (see Section \ref{quadri}).


\end{enumerate}
\end{prop}

\begin{proof}
Each assertion follows by direct computation. For conciseness, we provide only one such computation here. Consider successive actions of the seed involutions
$$
(x,y)\in {\mathcal P}(\lambda) \xmapsto{j_1} ({\tilde x},y)\in {\mathcal P}(-\lambda) \xmapsto{j_2}({\hat {\tilde x}},y)\in {\mathcal P}(\lambda) \xmapsto{k_1}({\hat {\tilde x}},{\tilde y})\in {\mathcal P}(-\lambda) \xmapsto{k_2} ({\hat {\tilde x}},{\hat {\tilde y}})\in {\mathcal P}(\lambda),
$$
where the tilde and hat annotation of the points indicate the images of horizontal and vertical switches as indicated in Figures \ref{fig1}--\ref{fig2}. This proves that $j_1\circ j_2=i_1,$\; $k_1\circ k_2=i_2,$ , from which it follows that
the QRT map can be written as $\phi=i_1\circ i_2=j_1\circ j_2\circ k_1 \circ k_2:(x,y)\in {\mathcal P}(\lambda)\mapsto ({\hat {\tilde x}},{\hat {\tilde y}})\in {\mathcal P}(\lambda),$ as asserted in item (i).
\end{proof}

\section{Integrable maps}


In this section, we consider outcomes from the factorisation of the QRT map in terms of seed involutions, given in Section 2.
Using the generators of the group $G$ we can construct  {\em words} of length 2 and 3 that correspond to the following maps:
\be \label{hky}
\Upsilon_i=j_i\circ j_{i+1}\circ k_i,\;\; \Phi_i=k_i\circ k_{i+1}\circ j_i,\;\; \Psi_i=j_i\circ k_i\;\; \Omega_i=j_i\circ k_{i+1}\;\; i=1,2,
\ee
where the indices are considered modulo $2$.
Note that  $\Upsilon_i$ and $\Phi_i$ are of HKY-type since they anti-preserve the integral $I(x,y),$  while $\Psi_i$ and $\Omega_i$  preserve it. Note also that the QRT map $(\phi)$ can be expressed in terms of $\Psi_i$ as
$$
\phi=i_2\circ i_1= k_1 \circ k_2 \circ j_1 \circ j_2= \Psi_2 \circ \Psi_1.
$$
Note that here we have achieved a refactorisation of a QRT map in terms of 4 involutions $j_i, k_i,$ $i=1,2$ or in terms of 2 non-commuting infinite order maps $\Psi_i.$
In terms of the generators of the group $G'$, the mappings (\ref{hky}) read
\be \label{hky'}
\begin{aligned}
\Upsilon_1&=(J\circ R\circ S)^2R \circ J \circ R,&\Phi_1&=(R\circ J\circ S)^2 \circ J,\\
\Upsilon_2&=(J\circ R\circ S)^2S \circ J \circ S,&\Phi_2&=(R\circ J\circ S)^2 \circ R \circ S \circ J \circ R \circ S,
\\
\Psi_1&=(J\circ R)^2,&\Omega_1&=(J\circ S)^2,\\
\Psi_2&=(R\circ S\circ J\circ S)^2,&\Omega_2&=(R\circ S \circ J \circ R)^2.
\end{aligned}
\ee
{  A mapping of the plane $\psi : (x, y) \mapsto (X,Y)$ is called (anti) measure-preserving with
density $m(x, y)$, if the Jacobian determinant ($J$) of $\psi$ can be written as $J =\frac{m(x,y)}{m(X,Y)}.$
 Anti measure-preservation
corresponds to a measure-preserving and orientation-reversing mapping  \cite{Roberts-1992}. A measure preserving map of the plane is Liouville integrable since it holds $m(x,y)dx\wedge dy= m(X,Y)dX\wedge dY$.
\begin{remark}
Mappings (\ref{hky}) and (\ref{hky'}) are measure or anti measure-preserving maps with density ${\displaystyle m(x,y)=({\bf X}^T\mathscr{A}_0{\bf
Y})^{-1}}$. Hence they are Liouville integrable maps. Moreover, all words of the group $G$ or $G'$ correspond to  Liouville integrable maps.
\end{remark}}

\begin{example}[The mapping $\Psi_1$]
Consider $\Psi_1=(J\circ R)^2,$  where $J\circ R$ itself is a non-QRT map that explicitly reads:
$$
  J\circ R:\left(x,y; \mathscr{A}_0\right)\mapsto \left(-{\dis \frac{2(\gamma x^2+\zeta x+\mu)+(\beta x^2+\epsilon x+\lambda)y}{\beta x^2+\epsilon x+\lambda+2(\alpha x^2+\delta x+\kappa)y}},x ;\mathscr{A}_0^T \right).
$$
Since, for this map, the parameters vary as $\mathscr{A}_0\mapsto \mathscr{A}_0^T,$ it follows that $\alpha, \epsilon, \mu$ are constants, while the remaining parameters satisfy
$$
\begin{array}{l}
\beta=c_1+c_2(-1)^n\\
\delta=c_1-c_2(-1)^n
\end{array},\;
\begin{array}{l}
\gamma=c_3+c_4(-1)^n\\
\kappa=c_3-c_4(-1)^n
\end{array},\;
\begin{array}{l}
\lambda=c_5+c_6(-1)^n\\
\zeta=c_5-c_6(-1)^n
\end{array},
$$ where $c_i, i=1,\ldots, 6$ constants.  Under this identification, we have the non-autonomous  map
\be \label{qv}
y_{n+1}=- \frac{2(\gamma y_n^2+\zeta y_n+\mu)+(\beta y_n^2+\epsilon y_n+\lambda)y_{n-1}}{\beta y_n^2+\epsilon y_n+\lambda+2(\alpha y_n^2+\delta y_n+\kappa)y_{n-1}}.
\ee
{  Alternating maps of QRT type were introduced in \cite{Quispel-2003}. Mapping (\ref{qv}) stands as a prototypical example of an alternating map of non-QRT (HKY) type.}
Setting $\mathscr{A}_0=\mathscr{A}_0^T,$ the mapping $(\ref{qv})$ becomes autonomous. It is exactly the $(1,1)$-reduction of Viallet's $Q_V$ integrable partial difference equation \cite{viallet-2009}. Further specialisation of the parameters leads to the $(1,1)$-reduction of the Adler's $Q4$ partial difference equation. This further reduction was obtained by \cite{joshi-2006,atkinson-2008,kassotakis-2010}. A non-autonomous version of the latter was recently placed in the geometric framework of
the Painlev\'e equations in \cite{Atk-Joshi-2016}.
\end{example}

\subsection{Quadrirational  maps} \label{quadri}

In his studies on Yang-Baxter maps and their interplay with the theory of geometric crystals \cite{et-2003}, Etingof introduced the notion of {\it non-degenerate rational maps}. Nowadays, Etingof's terminology is widely  used but with a different name. Instead of the term non-degenerate rational maps, the terminology {\it quadrirational maps} is used. In \cite{Kassotakis-2016} a natural extension of the notion of quadrirational maps to arbitrary dimensions was introduced. Here we show that the maps constructed in Proposition 2.3 (vii) are quadrirational.

\begin{dfn}
 A rational map  $ \mathbb{P}^1\times \mathbb{P}^1 \ni (x,y) \rightarrow    (X,Y) \in \mathbb{P}^1\times \mathbb{P}^1$ is called quadrirational if the maps $(X,Y)\rightarrow (x,y),$ $(X,y)\rightarrow (x,Y),$ and $(x,Y)\rightarrow (X,y)$ are all rational maps.
\end{dfn}
By direct computation, we find that mappings $\Omega_i\; i=1,2$  are quadrirational and they explicitly read:
$$
\begin{array}{l}
\Omega_1=j_1\circ k_2: (x,y)\mapsto(X,Y), \\ [3mm]
\mbox{where}\;\;\left\{\begin{array}{l}
X=-{\dis \frac{2(\kappa y^2+\lambda y+\mu)+(\delta y^2+\epsilon y+\zeta)x}{\delta y^2+\epsilon y+\zeta+2(\alpha y^2+\beta y+\gamma)x}},\\[3mm]
Y= -{\dis \frac{2(|\mathscr{A}_0|_{31} x^2+|\mathscr{A}_0|_{21} x+|\mathscr{A}_0|_{11})+(|\mathscr{A}_0|_{32} x^2+|\mathscr{A}_0|_{22} x+|\mathscr{A}_0|_{12})y}{|\mathscr{A}_0|_{32} x^2+|\mathscr{A}_0|_{22} x+|\mathscr{A}_0|_{12}+
2(|\mathscr{A}_0|_{33} x^2+|\mathscr{A}_0|_{23} x+|\mathscr{A}_0|_{13})y}},
     \end{array}\right.
     \end{array}
$$
$$
\begin{array}{l}
\Omega_2=j_2\circ k_1: (x,y)\mapsto(X,Y), \\ [3mm]
\mbox{where}\;\; \left\{ \begin{array}{l}
X=-{\dis \frac{2(|\mathscr{A}_0|_{13} y^2+|\mathscr{A}_0|_{12} y+|\mathscr{A}_0|_{11})+(|\mathscr{A}_0|_{23} y^2+|\mathscr{A}_0|_{22} y+|\mathscr{A}_0|_{21})x}{|\mathscr{A}_0|_{23} y^2+|\mathscr{A}_0|_{22} y+|\mathscr{A}_0|_{21}+
2(|\mathscr{A}_0|_{33} y^2+|\mathscr{A}_0|_{32} y+|\mathscr{A}_0|_{31})x}},\\ [3mm]
Y= -{\dis \frac{2(\gamma x^2+\zeta x+\mu)+(\beta x^2+\epsilon x+\lambda)y}{\beta x^2+\epsilon x+\lambda+2(\alpha x^2+\delta y+\kappa)y}}.
    \end{array}\right.
    \end{array}
$$
Note that there is $\Omega_1^2=\Omega_2^2=id$ and by
following Proposition \ref{prop1}, both $\Omega_1$ and $\Omega_2$ factorise:
\begin{align*}
\Omega_1&=j_1\circ k_2=J\circ S \circ J\circ S =(J\circ S)^2,\\
\Omega_2&=j_2\circ k_1=R\circ S \circ J \circ R\circ S \circ R \circ J \circ R=(S\circ R \circ J \circ R)^2.
\end{align*}
Clearly, we also have  $(J\circ S)^4=(S\circ R \circ J \circ R)^4=id,$ so $\omega_1:=J\circ S$ and
$\omega_2:=R\circ S \circ J \circ R$ are period 4 maps.

Both maps $\Omega_i$  are involutions and  preserve the same pencil of biquadratic curves.  Unlike the QRT mapping, whose singular points
is given by the set $\Sigma$ of the base points of the pencil, the mapping $\Omega_1$ has $\sigma_1$ as a set of singular points and the set $\sigma_2$  as
fixed points,
while the mapping $\Omega_2$  has $\sigma_2$ as the set of singular points and $\sigma_1$ as  the set of fixed points.  Because of this property we will call $\Omega_2$ the dual map of the map $\Omega_1$.

It was proven in \cite{ABSyb-2004,pap3-2010} that   quadrirational  maps are M\"obius equivalent to the  $F$-list of Yang-Baxter maps, so essentially $\Omega_i$ are Yang-Baxter maps. Moreover both maps $\Omega_i$ preserve the same integral $I(x,y)$ and they share the same set of parameters namely the matrix $\mathscr{A}_0$.
In Table 1, we present the parameter matrices for which  $\Omega_1$  (that is reffered to as $\mathscr{R}$) corresponds to the $F$-list \cite{ABSyb-2004} of quadrirational Yang-Baxter maps, as well as the maps  $\omega_1$ (that is reffered to as $r$) which satisfy $r^2=\mathscr{R}$. In Table 2, we present the corresponding values of the parameter matrices for the mappings $\Omega_2$ (that is reffered to as $\mathscr{L}$) and $\omega_2$ (that is reffered to as $l$).

Because no distinction is made between mappings $\mathscr{L}$  and the Yang-Baxter maps $\mathscr{R}$ though out our construction, the mappings $\mathscr{L}$ can be considered as constituting the {\em dual $F$-List}. Although $\mathscr{L}$ are not Yang-Baxter maps, nevertheless the common property shared by the $F$-list and  dual $F$-list is that both maps $r$ and $l$ satisfy the modified Yang-Baxter relation (\ref{byb}) namely:
$$
 (\mathscr{M}_{ij}\circ \mathscr{M}^{-1}_{ik} \circ \mathscr{M}_{jk})^4=id,
$$
where $\mathscr{M}^4=id$ and $\mathscr{M}$ is either the map $r$ or the map $l$.

\begin{example}[The $F_I$ Yang-Baxter map and its dual]
Choosing the parameter matrix $\mathscr{A}_0$ as
  $$
\mathscr{A}_0=\left( \begin{array}{lll}
            0&p-q&q(1-p)\\
            0&2p(q-1)&0\\
            p(1-p)&p(p-q)&0
            \end{array}
\right),$$
the integral reads
{\small $$
I(x,y)= \frac{(q-p)x^2y+q(p-1)x^2+p(p-1)y^2+2p(1-q)xy+p(q-p)y}
{(p-q)xy^2+q(q-1)x^2+p(q-1)y^2+2q(1-p)xy+q(p-q)x}.
$$}
In this case, the map $\Omega_1$  becomes exactly the $F_I$ Yang-Baxter map,  while  $\Omega_2$  is a new map, which we denote by ${\hat F_I}$ and describe as the dual of $F_I$. Their explicit expressions are given by
\begin{align*}
&\mathscr{R}: (x,y)\mapsto (X,Y),\\
&\qquad \mbox{where}\;\; \begin{cases}
X&=p y \mathscr{P} ,\;\;  Y=q x \mathscr{P} ,\\
\mathscr{P}&={\dis \frac{(1-q) x+q-p+(p -1)y}{q(1-p) x+(p-q)xy+p (q-1)y}},
\end{cases}&(F_I)\\
&\mathscr{L}:(x,y)\mapsto (X,Y),\\
&\qquad\mbox{where}\;\;\begin{cases}
X&=y+ W(x,y,p,q),\;\;
 Y=x+ W(y,x,q,p),\\
 W(x,y,p,q)&={\dis \frac{(q-p)\left(q(x+y-2 x y)+y^2(x+y-2)\right)}{q(p-q)2qx(q-1)+y\left(2q-2pq+(p-q)y\right)}},
\end{cases}& ({\hat F_I})
\end{align*}

As $ F_I$ is a Yang-Baxter map, it  satisfies:
$
\mathscr{R}^2=id, \quad (\mathscr{R}_{ij}\circ \mathscr{R}_{ik}\circ \mathscr{R}_{jk})^2=id.
$
We also have the maps $r$ and $l$ that correspond to the   $\omega_i, i=1,2$ hence they satisfy $r^2=\mathscr{R}$  and $l^2=\mathscr{L}.$ They read:
\[
 \begin{aligned}
& r:(x,y,p,q)\mapsto (X,Y,P,Q)\\
&\qquad \mbox{where}\;\;\left\{\begin{array}{ll}
X=qx {\dis \frac{(1-q) x+q-p+(p -1)y}{q(1-p) x+(p-q)xy+p (q-1)y}}, &  P=q,\\
Y=x,& Q=p,
\end{array}\right.\\
& l:(x,y,p,q)\mapsto (X,Y,P,Q)\\
&\qquad \mbox{where}\;\;\left\{\begin{array}{ll}
X=x +{\dis \frac{-(q-p)\left(p(x+y-2 x y)+x^2(x+y-2)\right)}{-p(p-q)2py(p-1)+x\left(2p-2pq-(p-q)x\right)}}, &  P=q,\\
Y=x, & Q=p.
\end{array}\right.
\end{aligned}
\]
It is natural to ask whether there is a common property that both maps share. The answer is yes,
both $r, l$ (recalling $\mathscr{R}=r^2, \mathscr{L}=l^2,$) satisfy  a modified Yang-Baxter relation:
\begin{equation} \label{yb2}
\mathscr{M}^4=id, \quad (\mathscr{M}_{ij}\circ \mathscr{M}^{-1}_{ik}\circ \mathscr{M}_{jk})^4=id.
\end{equation}

\end{example}

\begin{longtable}{|c|c|l|l|}
 \caption{The  $F$-list of Yang-Baxter maps}
 \\ \hline
   & $\mathscr{A}_0$  & $\mathscr{R}:(x,y)\mapsto (X,Y)$ & $r:(x,y,p,q)$\\
    &&&\ $\mapsto (X,Y,P,Q)$\\ [4mm]  \hline
  $F_I$ & $\left( \begin{array}{lll}
            0&p-q&q(1-p)\\
            0&2p(q-1)&0\\
            p(1-p)&p(p-q)&0
            \end{array}\right)$
            &
           $\begin{array}{l}
            X=p y W(x,y,p,q),\\
            Y=q x W(x,y,p,q),\\
             W(x,y,p,q)= \frac{(1-q) x+q-p+(p -1)y}{q(1-p) x+(p-q)xy+p (q-1)y}
           \end{array}$
    &  $\begin{array}{ll}
          X=q x W(y,x,q,p) ,\\
          Y=x,\\
          P=q,\\
          Q=p
        \end{array}$ \\
        \hline
  $F_{II}$ & $\left( \begin{array}{lll}
            0&0&p\\
            0&-2p&0\\
           q&(p-q)&0
            \end{array}\right)$
            & $\begin{array}{l}
              X=y/p W(x,y,p,q),\\
              Y=x/q W(x,y,p,q),\\
              W(x,y,p,q)= \frac{p x-q y+q-p}{x-y},
              \end{array}$
    &  $\begin{array}{ll}
         X=x/q W(y,x,q,p) , \\
         Y=x,\\
          P=q,\\
          Q=p
        \end{array}$ \\  \hline
  $F_{III}$ & $\left( \begin{array}{lll}
            0&0&p\\
            0&-2p&0\\
           q&0&0
            \end{array}\right)$
            & $\begin{array}{l}
            X=y/p W(x,y,p,q),\\
            Y=x/q W(x,y,p,q),\\
            W(x,y,p,q)= \frac{p x-q y}{x-y}
             \end{array}$
    &  $\begin{array}{ll}
        X=x/q W(y,x,q,p),\\
        Y=x,\\
        P=q,\\
        Q=p
\end{array}$ \\  \hline
  $F_{IV}$ & $\left( \begin{array}{lll}
            0&0&1\\
            0&-2&0\\
           -1&(p-q)&0
            \end{array}\right)$  &
            $\begin{array}{l}
            X=y W(x,y,p,q),\\
            Y=x W(x,y,p,q),\\
           W(x,y,p,q)= 1-\frac{p-q}{x-y}
            \end{array}$
    &  $\begin{array}{ll}
        X=x W(y,x,q,p) , \\
        Y=x,\\
        P=q,\\
        Q=p
        \end{array}$ \\  \hline
    $F_{V}$ & $\left( \begin{array}{lll}
            0&0&1\\
            0&-2&0\\
           1&0&q-p
            \end{array}\right)$
            & $\begin{array}{l}
               X=y+ W(x,y,p,q),\\
               Y=x+ W(x,y,p,q),\\
               W(x,y,p,q)= \frac{p-q}{x-y}
              \end{array}$
    &  $\begin{array}{ll}
       X=x+ W(y,x,q,p) ,\\
       Y=x,\\
       P=q,\\
       Q=p
\end{array}$ \\
  \hline
\end{longtable}
\begin{longtable}{|c|c|l|l|}
 \caption{The dual ${\hat F}$-list of Yang-Baxter maps}
 \\ \hline
   & $\mathscr{A}_0$  & $\mathscr{R}:(x,y)\mapsto (X,Y)$ & $r:(x,y,p,q)\mapsto (X,Y,P,Q)$  \\ [4mm]   \hline
  ${\hat F}_I$ & $\left( \begin{array}{lll}
            0&p-q&q(1-p)\\
            0&2p(q-1)&0\\
            p(1-p)&p(p-q)&0
            \end{array}\right)$
            &
           $\begin{array}{l}
            X=y+ W(x,y,p,q),\\
             Y=x+ W(y,x,q,p),\\
             W(x,y,p,q)=\\ \frac{(q-p)\left(q(x+y-2 x y)+y^2(x+y-2)\right)}{q(p-q)2qx(q-1)+y\left(2q-2pq+(p-q)y\right)}
           \end{array}$
    &  $\begin{array}{ll}
          X=x+ W(y,x,q,p) ,\\
          Y=x,\\
          P=q,\\
          Q=p
        \end{array}$ \\
        \hline
  ${\hat F}_{II}$ & $\left( \begin{array}{lll}
            0&0&p\\
            0&-2p&0\\
           q&(p-q)&0
            \end{array}\right)$
            & $\begin{array}{l}
X=y+ W(x,y,p,q),\\
 Y=x+ W(x,y,p,q) \\
  W(x,y,p,q)= \frac{(q-p)(x+y-2xy)}{p-q-2px+2qy}
              \end{array}$
    &  $\begin{array}{ll}
         X=x+W(y,x,q,p) , \\
         Y=x,\\
          P=q,\\
          Q=p
        \end{array}$ \\  \hline
  ${\hat F}_{III}$ & $\left( \begin{array}{lll}
            0&0&p\\
            0&-2p&0\\
           q&0&0
            \end{array}\right)$
            & $\begin{array}{l}
            X=qy W(x,y,p,q),\\
            Y=px W(x,y,p,q),\\
            W(x,y,p,q)= \frac{x-y}{p x-q y}
             \end{array}$
    &  $\begin{array}{ll}
        X=px W(y,x,q,p),\\
        Y=x,\\
        P=q,\\
        Q=p
\end{array}$ \\  \hline
  ${\hat F}_{IV}$ & $\left( \begin{array}{lll}
            0&0&1\\
            0&-2&0\\
           -1&(p-q)&0
            \end{array}\right)$  &
            $\begin{array}{l}
            X=y+W(x,y,p,q),\\
            Y=x+W(x,y,p,q),\\
           W(x,y,p,q)= \frac{(q-p)(x+y)}{p-q-2(x+y)}
            \end{array}$
    &  $\begin{array}{ll}
        X=x+W(y,x,q,p) , \\
        Y=x,\\
        P=q,\\
        Q=p
        \end{array}$ \\  \hline
    ${\hat F}_{V}$ & $\left( \begin{array}{lll}
            0&0&1\\
            0&-2&0\\
           1&0&q-p
            \end{array}\right)$
            & $\begin{array}{l}
               X=y+ W(x,y,p,q),\\
               Y=x+ W(x,y,p,q),\\
               W(x,y,p,q)= -\frac{p-q}{x-y}
              \end{array}$
    &  $\begin{array}{ll}
       X=x+ W(y,x,q,p) ,\\
       Y=x,\\
       P=q,\\
       Q=p
\end{array}$ \\
  \hline
\end{longtable}
\ \\
\ \\

\section{Discussion}
In this paper, we investigated the additional structure that arose  through a refactorisation of each of the QRT involutions which corresponded to a specific class of parameter matrices (\ref{matrices00}).  Liouville integrable maps of QRT  and HKY type, quadrirational Yang-Baxter map as well as some maps with novel dynamics $(\ref{byb})$  arose through our construction. All these  maps were considered as compositions of involutions,  which preserved, anti-preserved   or sent a biquadratic invariant to its reciprocal (see Proposition \ref{prop1}).

A main characteristic  of these  involutions is that they act non-trivially on the parameter matrices $(\mathscr{A}_0, \mathscr{A}_1)$ as well as on the coordinates $(x,y)$. They define groups $G, G' $  with $G\subset G'$  where $G'$ is the hyperbolic $(2,4,\infty)$ triangle Coxeter group  modulo some relations (see point (iii) of Proposition \ref{prop1}).

Inspired by \cite{viallet-1991a,viallet-1991b}, we end this discussion  by extending further the group $G'$ adding the involutions:
$$
H: (x,y;\mathscr{A}_0)\mapsto \left(1/x,1/y;1/\mathscr{A}_0\right),\quad M:(x,y;\mathscr{A}_0)\mapsto (x,y; |\mathscr{A}_0| ), \quad L: (x,y)\mapsto \left(-1/x,-1/y;\Sigma \mathscr{A}_0 \Sigma\right),
$$
where $|\mathscr{A}_0|$ is the matrix  of the minors of  $\mathscr{A}_0,$  ${\dis \Sigma=  \left(\begin{array}{ccc}
                                                                              0&0&1\\
                                                                              0&-1&0\\
                                                                              1&0&0
                                                                              \end{array}\right)} $
                                                                               and $1/\mathscr{A}_0:= {\dis  \left(\begin{array}{ccc}
1/\alpha& 1/\beta &1/\gamma\\
1/\delta& 1/\epsilon& 1/\zeta\\
1/\kappa& 1/\lambda & 1/\mu
\end{array}\right) }  $  is the Hadamard inverse of $\mathscr{A}_0$. Then is is easy to show that
\[{\dis N=R\circ M \circ L:(x,y;\mathscr{A}_0)\mapsto (-1/y,-1/x,\mathscr{A}_0^{-1})},\]
where $\mathscr{A}_0^{-1}$ the inverse matrix of $\mathscr{A}_0.$

Finally, the mappings $\sigma_1=N \circ H,\sigma_2= M\circ H, \sigma_3=N\circ R\circ H,$ neither commute nor satisfy $I\circ \sigma_1=I\circ \sigma_2=I\circ \sigma_3$. They  define nontrivial  integrable dynamics on ${\mathbb P}^8$ if we consider only their action on $\mathscr{A}_0$. Moreover  Kontsevich's periodicity identity \cite{kont-2011,iyudu-2013} can be verified, namely:
$$
(\sigma_3)^3: \mathscr{A}_0\mapsto \Delta_L \mathscr{A}_0 \Delta_R,
$$
where
$$
\Delta_L=\frac{1}{det(\mathscr{A}_0) det(1/\mathscr{A}_0)} \left(\begin{array}{ccc}
                                                               -\alpha \beta \gamma \kappa_1 \lambda_1 \mu_1 &0&0\\
                                                                 0&\delta \epsilon \zeta \delta_1 \epsilon_1 \zeta_1&0\\
                                                                 0&0& \alpha_1 \beta_1 \gamma_1 \kappa \lambda \mu
                                                                 \end{array}\right),
$$
$$
\Delta_R= \frac{1}{det(\mathscr{A}_0) det(1/\mathscr{A}_0)} \left(\begin{array}{ccc}
                                                               \alpha \delta \kappa \gamma_1 \zeta_1 \mu_1 &0&0\\
                                                                 0&-\beta \epsilon \lambda \beta_1 \epsilon_1 \lambda_1&0\\
                                                                 0&0& \alpha_1 \delta_1 \kappa_1 \gamma \zeta \mu
                                                                 \end{array}\right).
$$

{  To conclude, we have shown that the factorisation condition on the QRT involutions is equivalent to the conditions (\ref{n=2-cond}).
 It would be interesting to investigate the remaining conditions of (\ref{lamma1}) and investigate the algebraic structure of the associated maps. When conditions of our Lemma (\ref{lamma1}) are satisfied and placed inside the framework of the general theory recently introduced in \cite{rob-2015}, we expect important classes of biquadratics to arise, which suggests an enrichment of the latter theory.

 Also, our approach might cast some light on the exotic behavior of mappings that map a pencil of biquadratic curves to higher degree curves and return to the original pencil of curves after some iterations. Finally, an open direction of research is to show how the re-factorisation  procedures presented in this manuscript can be extended to novel Liouville integrable maps obtained in higher dimensions \cite{Joshi-Viallet:2018,Joshi:2019}.
}

\subsubsection*{Acknowledgments}
The authors would like to thank J. Atkinson, M. Noumi and J.A.G. Roberts for valuable discussions. PK was supported in part
by the Australian Research Council Discovery Grant No. DP110102001 and NJ by an ARC Georgina Sweet Australian Laureate Fellowship FL120100094.


\end{document}